\documentclass[preprint,11pt,authoryear]{elsarticle}

\usepackage{amsmath}
\usepackage{float}
\usepackage{geometry}
\usepackage{array}
\usepackage{booktabs}
\usepackage{graphicx}
\usepackage{xcolor}
\usepackage{caption}
\usepackage{amsthm}
\usepackage{algorithm}
\usepackage{algpseudocode}
\usepackage{caption}
\usepackage{subcaption}
\usepackage{natbib}
\newtheorem{theorem}{Theorem}
\newtheorem{lemma}{Lemma}
\usepackage{hyperref}
\hypersetup{
    colorlinks=true,
    linkcolor=blue,
    filecolor=blue,      
    urlcolor=blue,
    citecolor=blue
}
\numberwithin{equation}{section}
\newcounter{assump} 
\renewcommand{\theassump}{AN\arabic{assump}}

\newenvironment{assumption}[1][]{
  \refstepcounter{assump}%
  \par\noindent
  {\theassump.}
}{\par\normalfont}
 \usepackage[T1]{fontenc}
\makeatletter \oddsidemargin 0in \evensidemargin 0in \textwidth 16cm
\usepackage{amssymb}
\usepackage{amsmath}

\begin{document}

\begin{frontmatter}



\title{\textbf{Asymptotically Optimal Sequential Confidence Interval for the Gini Index Under Complex Household Survey Design with Sub-Stratification}}


\author[1]{Shivam} 
\author[2]{Bhargab Chattopadhyay}
\cortext[cor1]{Corresponding author}
\author[1]{Nil Kamal Hazra\corref{cor1}}
\ead{nilkamal@iitj.ac.in}
\affiliation[1]{organization={Department of Mathematics},Department and Organization
            addressline={Indian Institute of Technology Jodhpur}, 
            city={Karwar},
            postcode={342037},
            country={India}}
            
\affiliation[2]{organization={School of Management and Entrepreneurship},Department and Organization
            addressline={Indian Institute of Technology Jodhpur}, 
            city={Karwar},
            postcode={342037},
            country={India}}
\date{}

\begin{abstract}
We examine the optimality properties of the Gini index estimator under complex survey design involving stratification, clustering, and sub-stratification. While Darku et al. (Econometrics, \textbf{8}, 26, 2020) considered only stratification and clustering and did not provide theoretical guarantees, this study addresses these limitations by proposing two procedures—a purely sequential method and a two-stage method. Under suitable regularity conditions, we establish uniform continuity in probability for the proposed estimator, thereby contributing to the development of random central limit theorems under sequential sampling frameworks. Furthermore, we show that the resulting procedures satisfy both asymptotic first-order efficiency and asymptotic consistency. Simulation results demonstrate that the proposed procedures achieve the desired optimality properties across diverse settings. The practical utility of the methodology is further illustrated through an empirical application using data collected by the National Sample Survey agency of India.

\end{abstract}


\end{frontmatter}





\let\WriteBookmarks\relax
\def\floatpagepagefraction{1}
\def\textpagefraction{.001}

\section{Introduction}\label{1.S1}
Economic inequality within a state or country primarily arises from the unequal distribution of income among individuals. To address this disparity, policymakers often design economic strategies informed by the underlying income distribution. Over the years, numerous measures of inequality have been proposed and extensively studied in the literature. Among these, the Gini index is one of the most used statistical measures for quantifying economic inequality in a population. 

Suppose $X$ is a random variable which represents an individual's income in a given population. Then, $G_X$, the population Gini index, is expressed as
\begin{eqnarray}\label{1.GL}
G_X=1-2\int_{0}^{1}\phi(p)dp.
\end{eqnarray}
Here $ \phi(\cdot)$ is the Lorenz function given by
\[\phi(p)=\frac{\alpha(p)}{\mu},\qquad \alpha(p)=\int_{0}^{p}z(u)du,\]
and 
$z(u)=\inf\{x: u\leq F(x)\}, u\in [0,1],$
is the $u$-th quantile of a population and $F(\cdot)$ being the c.d.f. (cumulative distribution function) of the random variable, $X$.
For a comprehensive review on Gini index and problems related to it, we refer the reader to  \cite{Delta2023}, \cite{IKS2025}, \cite{SIM}, \cite{Mukho21}, and the references therein. For other nonparamteric inequality measures and their characteristics, we refer to \cite{jokiel2025nonparametric}. \cite{sang2020depth} estimated the Gini index by solving some non-smooth U-statistic based estimating equations using a weighted jackknife empirical likelihood method. \cite{egozcue2019compositional} proposed an inequality index based on  Aitchison norm and compared it with Gini and Atkinson inequality indices using the mean household gross income per capita data from autonomous communities in spain.\\
To estimate the Gini index, data are typically collected through sample survey drawn from the target population. 
Several inference procedures for the Gini index under SRS have been proposed, including those by \cite{DARKU23}, \cite{SKDBMRPE}, \cite{SKDBGINI}, and \cite{dvd2009}, among others. However, real-world populations are often heterogeneous and geographically dispersed, making SRS inadequate. Consequently, many national statistical agencies employ complex household survey designs to capture more accurate representations of population characteristics. For example, the Household Consumption Expenditure Survey\footnote{\url{http://microdata.gov.in/NADA/index.php/catalog/CEXP/?page=1\&sort\_order=desc\&ps=15\&repo=CEXP$}}, conducted by the National Sample Survey (NSS) Organization in India, collects detailed data on household consumption expenditure across a wide range of items, recorded in both quantity and value terms. Similarly, the Land and Livestock Holding Survey\footnote{\url{https://microdata.gov.in/NADA/index.php/catalog/LLS/?page=1&sort_order=desc&ps=15&repo=LLS}} provides information on operational land holdings and livestock ownership in rural areas. These surveys serve distinct purposes and are essential for capturing key economic and social indicators. The survey designs incorporated in these studies typically involve stratification, clustering, and sub-stratification to better capture population heterogeneity. Consequently, such surveys are commonly referred to as complex household surveys. A widely used complex survey design, as employed by \cite{BD2005}, stratifies the population and then forms clusters of households within each stratum. Unlike the SRS scheme, these designs assign unequal selection probabilities to individuals and yield more accurate estimates. For details on the survey design used in complex household surveys used by different survey agencies, along with the corresponding limit theorem, we refer the reader to \cite{BD2007}.
For a complex household survey design, \cite{DARKU} addressed the construction of a confidence interval for Gini index with bounded-width under such a design. To account for the precision and accuracy of the such a confidence interval, \cite{DARKU} used sequential analysis approaches which adaptively determined the optimal number of clusters thereby attaining a confidence interval for Gini index with bounded-width. Using simulations they demonstrated that, for sufficiently small widths, the final number of clusters closely matches the optimal on average, and proved that the purely sequential procedure guarantees the pre-specified interval width is not exceeded. 
For a detail review on sequential analysis and their applications, we refer to \cite{tayanagi2024change}, \cite{mukho08}, \cite{ham2023hypothesis}, \cite{navarro2021high} and the references therein.
\subsection{Contribution of this article}
While \cite{DARKU} proposed two-stage and purely sequential procedures for finding confidence intervals for the Gini index with bounded-width under complex survey design, their framework does not account for sub-stratification of clusters within each stratum, a feature common in practical household surveys. Moreover, they did not establish key asymptotic properties, such as asymptotic first-order efficiency and asymptotic consistency, for their procedures. These properties are critical to ensure that the final cluster size obtained via sequential sampling is nearly optimal on average while also achieve the desired confidence interval with controlled precision and desired accuracy.

Motivated by these gaps, this paper makes several key contributions. We propose purely sequential and two-stage procedures for finding confidence intervals for Gini index with bounded-width under complex household survey design that incorporate sub-stratification of clusters within each stratum, based on the economic conditions of households. Under suitable regularity conditions, we establish uniform continuity in probability for the proposed estimator under this design, thereby contributing to the development of random central limit theorems. Furthermore, we show that the resulting procedures satisfy both asymptotic first-order efficiency and asymptotic consistency. Theoretical findings are validated through extensive simulation studies and an empirical application based on data collected by India’s official survey agency, the National Sample Survey (NSS). Additionally, we examine the impact of sub-stratification on the estimation of the Gini index and its associated standard error.

The rest of this paper is presented as follows: Section \ref{1.S2} describes the complex household survey design used in this paper. Further,  we obtain the Gini index estimator under the survey design. Also, we introduce the list of notation and the set of assumptions used in subsequent sections. In Section \ref{1.S3}, we propose both purely sequential and two-stage procedures for obtaining a confidence interval for the Gini index with bounded-width. In Section \ref{1.S4}, we discuss the efficiency and the consistency properties of the procedures discussed in Section \ref{1.S3}. In Section \ref{1.S5}, we do several simulation studies to validate the theoretical findings of this paper. Subsequently, a real data analysis is done in Section \ref{1.S6}. {In Section \ref{1.S7}, we investigate how sub-stratification influences the final number of clusters and the accuracy of Gini index estimation.} Finally, the concluding remarks are given in Section \ref{1.S8}. For improved readability, all proofs of the theorems along with all supporting lemmas are deferred to the  Appendix.
\section{Complex survey design and Estimator of Gini index}\label{1.S2}
Assume that there are $ S $ distinct strata in the population which is indexed by $ s = 1, 2, \ldots, S $. Each of the $ s^{\text{th}} $ strata is further divided into $ H_s $ clusters, indexed by $ c_s = 1, 2, \ldots, H_s $. Thus, the total number of clusters across all strata is $H = \sum_{s=1}^{S} H_s$. Additionally, each cluster within a stratum is subdivided into two sub-strata based on household affluence: affluent ($ b_{c_s} = 1 $) and non-affluent ($ b_{c_s} = 2 $). Further, suppose within the $ c_s^{\text{th}} $ cluster of the $ s^{\text{th}} $ stratum, the $ b_{c_s}^{\text{th}} $ sub-stratum contains $ M_{sc_sb_{c_s}} $ households. Therefore, each cluster contains a total of $M_{sc_s} = \sum_{b_{c_s}=1}^{2} M_{sc_sb_{c_s}} \quad \text{households}$.

For estimation, suppose that from the $ s^{\text{th}} $ stratum, a sample of $ n_s $ clusters is drawn using probability proportional to size sampling with replacement, where size refers to the number of households. Within each selected cluster, $k$ households are selected at random from each sub-stratum using simple random sampling. The number of sampled clusters is
\[n = \sum_{s=1}^{S} n_s, \quad \text{with} \quad \frac{n_s}{n} = \frac{H_s}{H} = a_s.
\]

Let $ x_{sc_sb_{c_s}h} $ denote the observed value for the $ h^{\text{th}} $ household in sub-stratum $ b_{c_s} $ of cluster $ c_s $ in stratum $ s $ (e.g., household income or expenditure). Each sampled household is assigned a weight $ W_{sc_sb_{c_s}h} $, which is the inverse of its probability of inclusion in the sample. This inclusion probability is the product of the probability of including cluster $ c_s $ from stratum $ s $ $(p_{c_s} )$ and the probability of selecting household $ h $ from sub-stratum $ b_{c_s} $ of the selected cluster $( p_{b_{c_s}h})$. Since we have sample $n_s$ clusters with probability proportional to the size with replacement and $k$ households by simple random sampling, the values of $p_{c_s}$ and $p_{b_{c_s}h}$ are as follows:
\[p_{c_s}=\frac{n_sM_{sc_s}}{\sum_{c_s=1}^{H_s}M_{sc_s}},  p_{b_{c_s}h}=\frac{k}{M_{sc_sb_{c_s}}}.\]
Therefore, the sampling weight is given by
\[
W_{sc_sb_{c_s}h} = \frac{1}{p_{c_s} \cdot p_{b_{c_s}h}}=\frac{(\sum_{c_s=1}^{H_s}M_{sc_s})M_{sc_sb_{c_s}}}{n_sM_{sc_s}k}.
\]
Therefore, the total of per-household weights is 
\[W=\sum_{s=1}^{S}\sum_{c_s=1}^{n_s}\sum_{b_{c_s}=1}^{2}\sum_{h=1}^{k}W_{sc_sb_{c_s}h}.\]
Define normalized weight as 
\[w_{sc_sb_{c_s}h}=W^{-1}W_{sc_sb_{c_s}h}.\] 
Using the normalized weights $ w_{sc_sb_{c_s}h} $, we now define the unbiased estimators for key population quantities under the complex design. The unbiased estimator of $\mu$ and the cumulative distribution function $F(x)$ are respectively given by
\[\widehat{\mu}_{n}=\sum_{s=1}^{S}\sum_{c_s=1}^{n_s}\sum_{b_{c_s}=1}^{2}\sum_{h=1}^{k}w_{sc_sb_{c_s}h}x_{sc_sb_{c_s}h},\text{  and      }\widehat{F}_{n}(x)=\sum_{a=1}^{S}\sum_{d=1}^{n_a}\sum_{b_{d}=1}^{2}\sum_{c=1}^{k}w_{adb_{d}c}\textbf{1}(x_{adb_{d}c}\leq x).\]
Based on our design, an estimator of the Gini index $G_X$ is
\begin{eqnarray}\label{1.3}
\widehat{G}_n = 1 - \frac{2}{\widehat{\mu}_n} \sum_{s=1}^{S} \sum_{c_s=1}^{n_s}\sum_{b_{c_s}=1}^{2} \sum_{h=1}^{k}w_{sc_sb_{c_s}h} x_{sc_sb_{c_s}h} \left(1 - \widehat{F}_n(x_{sc_sb_{c_s}h}) \right).
\end{eqnarray}
Alternatively, using the Lorenz curve formulation, the Gini index may be expressed as:
\[\widehat{G}_n=1-2\int_{0}^{1}\widehat{\phi}_n(p)dp,\]
where for $p \in [0,1],$  $ \widehat{\phi}_n(.)$ is the estimator of Lorenz function is given by 
\[\widehat{\phi}_n(p)=\frac{\widehat{\alpha}_n(p)}{\widehat{\mu}_n},\qquad \widehat{\alpha}_n(p)=\int_{0}^{p}\widehat{z}_n(u)du,\]
where for $p\in[0,1]$, 
\[\widehat{z}_n(p)=\text{inf}_{s,c_s,b_{c_s},h}\{x_{sc_sb_{c_s}h}; \widehat{F}_n(x_{sc_sb_{c_s}h})\geq p\}\]

is the $p^{th}$ sample quantile.

We now state the asymptotic distribution of the Gini estimator under the complex survey design. This is essential for constructing a bounded-width confidence interval with prescribed accuracy. Before proceeding with finding a bounded-width confidence interval, we present the notations and assumptions that will be used throughout the paper.

\noindent\textbf{Notations:}\\
N1. For each $p\in[0,1]$, define $\theta=(\mu,\alpha(p),z(p))$ is the population parameter. Let $\theta_0=(\mu_0,\alpha_0(p),z_0(p))$ is the true value of population parameter, and $\Theta$ is the parameter space.\\
N2. \begin{align}
\Tilde{\mathbf{m}}(x, \theta) &=
\begin{pmatrix}
\Tilde{m}^1(x, \theta) & 
\Tilde{m}^2(x, \theta) & 
\Tilde{m}^3(x, \theta)
\end{pmatrix},\notag
\end{align}
where, $\Tilde{{m}}^1(x,\theta)=\mu-x,\Tilde{{m}}^2(x,\theta)=\alpha(p)-x\textbf{1}(x\leq z(p)),\text{ and }\Tilde{{m}}^3(x,\theta)=p-\textbf{1}(x\leq z(p))$.\\
N3. $\Tilde{\textbf{m}}_i(\theta)=\sum_{s=1}^{S}\frac{1}{a_s}\textbf{1}(s_i=s)\frac{(\sum_{c_{s_i}=1}^{H_{s_i}}M_{s_i{c_{s_i}}})}{M_{s_ii}k}\sum_{b_{i}=1}^{2}M_{s_i{i}b_{i}}\sum_{h=1}^{k}\Tilde{\textbf{m}}(x_{s_iib_{i}h},\theta).$\\

\noindent\textbf{Assumptions:}\\
\begin{assumption}\label{1.AN1}
   $\Tilde{{m}}^j(\cdot,\theta)$ is continuous at each $\theta$ with probability 1, for each $j=1, 2, 3$.
\end{assumption}
\begin{assumption}\label{1.AN2}
     There exists $d(\cdot)$ with $E(d(\cdot))<\infty$ such that $|\Tilde{{m}}^j(t,\theta)|\leq d(t)$ for each $j=1,2,3$ for all $t$. 
\end{assumption}
\begin{assumption}\label{1.AN3}
 The parameter space $\Theta$ is compact. 
\end{assumption}
\begin{assumption}\label{1.AN4}
    $E(\Tilde{\textbf{m}}(x,\theta))$ is continuously differentiable at $\theta_0$ and $\frac{1}{n}\sum_{i=1}^{n}\frac{\partial}{\partial\theta}E(\Tilde{\textbf{m}}_i(\theta))\xrightarrow{a.s.} K,$ where $K$ is non-singluar matrix. 
\end{assumption}
\begin{assumption}\label{1.AN5}
     The sequence $v_n(\theta)=\frac{1}{\sqrt{n}}\sum_{i=1}^{n}\{\Tilde{\textbf{m}}_i(\theta)-E(\Tilde{\textbf{m}}_i(\theta))\}$ is stochastically equicontinuous.
\end{assumption}
\begin{assumption}\label{1.AN6} $\text{sup}_{\theta\in\Theta}E|\Tilde{\textbf{m}}_i(\theta)|^3<\infty.$
\end{assumption}
\begin{assumption}\label{1.AN7}
$\text{lim}_{n\rightarrow\infty}\sum_{i=1}^{n}\frac{Var(\Tilde{\textbf{m}}_i(\theta))}{i^2}<\infty \text{ and } \text{lim}_{n\rightarrow\infty}\frac{1}{n}\sum_{i=1}^{n}{Var(\Tilde{\textbf{m}}_i(\theta))}=W_0$. 
\end{assumption}
\begin{assumption}\label{1.AN9}
    Let $s,s'=1,2,..,S, c_s=1,2,..,n_s, c'_{s'}=1,2,..,n_{s'}, b_{c_s}, b_{c'_{s'}}=1,2 \text{ and } h,h'=1,2,..,k .\\ \text{ If } s\neq s' \text{ or } c_s\neq c'_{s'} , \{x_{sc_sb_{c_s}h}\} \text{ and }  \{x_{s'c'_{s'}b_{c'_{s'}}h'}\}$ are independent unless they are dependent. 
\end{assumption}
\begin{assumption}\label{1.AN10}
     For $s\neq s', \{x_{sc_sb_{c_s}h}\} \text{ and }\{x_{s'c'_{s'}b_{c'_{s'}}h'}\}$ are not necessarily identically distributed. 
\end{assumption}
\begin{assumption}\label{1.AN11}
     The variables at the cluster level across different clusters within each stratum are identically distributed. 
\end{assumption}
\begin{assumption}\label{1.AN12}
    The variance of strata level variables is finite for each stratum.
\end{assumption}
\begin{assumption}\label{1.AN13}
    The cluster totals have absolute $2+\delta$ raw moments $(\delta>0)$ which are uniformly bounded by $B_\delta>0$.
\end{assumption}
\begin{assumption}\label{1.AN14}
     For each $p\in[0,1], F(x) \text{ and } \frac{dF(x)}{dx}$ are continuous for all $x$ in neighbourhood $J$ of $p^{th}$ population quantile $z(p)$. 
\end{assumption}
\begin{assumption}\label{1.AN15}
     For some $0<C<\infty$,
$Var\{\widehat{F}_n(x+\delta)-\widehat{F}_n(x)\}\leq Cn^{-1}|\delta|$
for all $n$ and for $x \text{ and } x+\delta \text{ in } J$. 
\end{assumption}
\begin{assumption}\label{1.AN16}
 $\mathbb{E}[|X| \mid s] < \infty$.
\end{assumption}
In what follows, we below give statistical  interpretations of the above-mentioned assumptions. Assumption \ref{1.AN1} implies that small perturbations in the parameter $\theta$ lead to correspondingly small changes in $\tilde{m}^j(\cdot, \theta)$. Assumption \ref{1.AN4} ensures that $E(\tilde{\mathbf{m}}(x, \theta)) = 0$ holds only when $\theta = \theta_0$, meaning the moment conditions are satisfied uniquely at the true parameter value. Assumptions \ref{1.AN1}--\ref{1.AN4} together guarantee the consistency of the estimator $\hat{\theta}$, ensuring that it converges in probability to the true value $\theta_0$. To maintain stability in the long-run variability of the estimating functions, assumption \ref{1.AN7} is required. Consider a population partitioned into two strata, namely, Rural and Urban. The Rural stratum is further divided into clusters representing villages, while the Urban stratum consists of clusters corresponding to city blocks. Since rural and urban incomes are shaped by distinct economic environments, employment structures, and cost-of-living conditions, it is reasonable to treat them as independent. Likewise, incomes from different villages or city blocks are assumed to be independent due to differing local economic conditions, whereas individuals within the same village or block experience similar environments, creating intra-cluster dependence. This motivates the assumption of dependence within clusters or strata but independence across them, as stated in assumption \ref{1.AN9}. As household income distributions vary across strata due to demographic, geographic, and socioeconomic differences, it is reasonable to assume that incomes from different strata are not identically distributed, corresponding to assumption \ref{1.AN10}. In stratified sampling, strata are typically formed to be internally homogeneous with respect to key characteristics such as income, region, or occupation. Clusters within the same stratum (for example, villages or census blocks) represent comparable population segments and are therefore expected to have identically distributed cluster-level variables, as reflected in assumption \ref{1.AN11}. Taken together, assumptions \ref{1.AN1}--\ref{1.AN11} establish that, for sufficiently large cluster sizes, the estimator $\hat{\theta}$ is asymptotically normally distributed. Asymptotic normality is very useful in constructing confidence interval for the true parameter value. Assumption \ref{1.AN14} ensures that the quantile function is well-defined and exhibits smooth behavior, while assumption \ref{1.AN15} guarantees control over random fluctuations in the empirical distribution function within small intervals. In brief, these assumptions are essential because they ensure the validity of applying the law of large numbers, central limit theorem, and Taylor series expansion in the analysis, enabling consistent estimation and reliable inference.

Under the assumptions \ref{1.AN1}-\ref{1.AN11} and \ref{1.AN16} stated above, following \cite{BD2007}, as $ n_s \to \infty $ at the same rate for all strata, we have
\[\sqrt{n}(\widehat{G}_n-G_X)\xrightarrow{d}N(0,\xi^2),\]
where $\xi^2$ denotes the asymptotic variance of $\sqrt{n}\widehat{G}_n$. Using the asymptotic distribution, we construct a $100(1-\alpha)\%$ confidence interval for the population Gini index $G_X$ whose accuracy is a pre-assigned number $\alpha\in(0,1)$ given by
\begin{eqnarray}\label{1.4}
P\left(\widehat{G}_n-z_{\frac{\alpha}{2}}\frac{\xi}{\sqrt{n}}<G_X<\widehat{G}_n+z_{\frac{\alpha}{2}}\frac{\xi}{\sqrt{n}}\right)\geq1-\alpha.
\end{eqnarray}
To construct the bounded-width confidence interval, we fix the width of the confidence interval to $\omega$, similar to \cite{DARKU}, we have,
\begin{align}
    L=2z_\frac{\alpha}{2}\frac{\xi}{\sqrt{n}}\leq\omega \implies \frac{\omega\sqrt{n}}{2\xi}\geq z_{\frac{\alpha}{2}}\implies n\geq\frac{{4z_{\frac{\alpha}{2}}^2{\xi}^2}}{\omega^2}\equiv C
\end{align}
We note that $z_{{\alpha}/{2}}$ is the $100(1-{\alpha}/{2})^{th}$ percentile of $N(0,1)$ and $C=\lceil{{4z_{\alpha/2}^2}{\xi}^2}\omega^{-2}\rceil$ is the optimal total number of clusters to be sampled from all strata. Thus to ensure that the confidence interval width does not exceed $ \omega $, we need to collect data from at least $C$ clusters. Also, the optimal allocation of clusters to the $s^{th}$ stratum is $C_s=Ca_s$, where $a_s$ is the allocation proportion. If $C$ is known, the $100(1-\alpha)\%$ confidence interval for $G_X$, with width not exceeding pre-assigned confidence interval bound $\omega$ is
\[\Bigg(\widehat{G}_C-z_{\frac{\alpha}{2}}\frac{\xi}{\sqrt{C}},\widehat{G}_C+z_{\frac{\alpha}{2}}\frac{\xi}{\sqrt{C}}\Bigg).\]
Without knowing the distribution of income, the value of $C$ is unknown; in particular, the value of $\xi^2$ is unknown. 
Since $C$ is unknown, one must survey from the clusters at least in two stages to construct a bounded-width confidence interval achieving the desired coverage probability at least approximately. For that we require a consistent estimator of $\xi^2$. Based on \cite{BINDER}, a consistent estimator of $\xi^2$ is 
\begin{eqnarray}\label{V.E.}
V_n^2=n\sum_{s=1}^{S}\frac{n_s}{n_s-1}\sum_{c_s=1}^{n_s}\left(u_{sc_s}-\bar{u}_s\right)^2.
\end{eqnarray}
where,
\begin{align}
\bar{u}_s &= \frac{1}{n_s} \sum_{c_s=1}^{n_s} u_{sc_s}, \quad \text{and } 
u_{sc_s}= \frac{2}{\widehat{\mu}} \sum_{b_{c_s}=1}^{2} \sum_{h=1}^{k} w_{sc_sb_{c_s}h} \Bigg[ 
\left( \widehat{F}(x_{sc_sb_{c_s}h}) - \frac{\widehat{G}_n + 1}{2} \right) x_{sc_sb_{c_s}h} \notag\\
&\,\,\, \quad\quad\quad\quad\quad\quad\quad\quad + \sum_{a=1}^{S} \sum_{d=1}^{n_a} \sum_{b_d=1}^{2} \sum_{c=1}^{k} 
w_{adb_dc} x_{adb_dc} \mathbf{1}(x_{adb_dc} \geq x_{sc_sb_{c_s}h})  - \frac{\widehat{\mu}}{2} (\widehat{G}_n + 1)\Bigg].\label{std2.3}
\end{align}
For a comprehensive review of variance estimators of Gini index, we refer to \cite{langel2013variance}.
We begin by constructing confidence intervals using arbitrarily fixed cluster size $n$. The confidence interval is given by  
\[
\Bigg(\widehat{G}_n - z_{\frac{\alpha}{2}}\frac{V_n}{\sqrt{n}}, \; \widehat{G}_n + z_{\frac{\alpha}{2}}\frac{V_n}{\sqrt{n}}\Bigg).
\]  
To examine the adequacy of such fixed cluster size choice, a simulation study is carried out using three synthetic populations generated from Gamma(2.649, 0.84), Pareto(20{,}000, 5), and lognormal(2.185, 0.562) distributions. For each selected cluster size $n$, the width of the constructed confidence interval is computed and averaged over 5000 repetitions. The results are summarized in Table \ref{1.T19}. The first column specifies the income distribution, while the second and third columns correspond to the fixed cluster sizes $n_1$ and $n_2$. The fifth and sixth columns display the average widths of the respective confidence intervals, denoted by $\overline{w}_{n_1}$ and $\overline{w}_{n_2}$. The eighth and ninth columns show the proportions of intervals whose widths exceed 0.015 for $n_1$ and $n_2$, respectively. The last column shows the proportions of intervals whose widths exceed 0.01 for $n_2$. From Table \ref{1.T19}, it is apparent that the cluster size $n_1$ produces intervals wider than the desired width of $0.015$, indicating insufficient precision, whereas the cluster size $n_2$ results in intervals that are much narrower. Moreover, for a stricter target width of $0.01$, even $n_2$ fails to meet the requirement. Now, consider selecting a cluster size $n_3$ that lies between $n_1$ and $n_2$. The results indicate that the fixed cluster size approach may not efficiently achieve the desired level of precision. Although the cluster size $n_2$ attains the target width of 0.015, a moderately smaller cluster size $n_3$ can also reach the required precision with high probability, thereby reducing the overall sampling cost. This outcome highlights a fundamental limitation of the fixed cluster size method. Since the optimal cluster size is not known in advance, it often results in either under-sampling or over-sampling. Hence, there is a strong reason for adopting an adaptive strategy, such as a sequential procedure, that can determine the cluster size dynamically to meet the desired precision more effectively.
\begin{table}[hbt!]
\caption{Simulation results under fixed cluster sizes $n_1$, $n_2$ and $n_3$.}
\centering
\resizebox{\textwidth}{!}{%
\begin{tabular}{ccccccccccc}
		\toprule
		Distribution & $n_1$&$n_2$ &$n_3$&$\overline{w}_{n_1}$&$\overline{w}_{n_2}$&$\overline{w}_{n_3}$&$\widehat{P}(w_{n_1}>0.015)$&$\widehat{P}(w_{n_2}>0.015)$&$\widehat{P}(w_{n_3}>0.015)$&$\widehat{P}(w_{n_2}>0.01)$\\
        \midrule
		 Gamma(2.649,0.84)&750&1500&1000&0.0163&0.0116&0.0142&0.9902&0.0&0.0354&1.0\\
          Pareto(20000,5)&350&700&500&0.0172&0.0122&0.0144&0.8008&0.0448&0.3206&0.9792\\
           lognormal(2.185,0.562)&900&1700&1200&0.0159&0.0116&0.0138&0.852&0.0&0.0586&1.0\\
      \bottomrule
	\end{tabular}%
    }
    \label{1.T19}
\end{table}
\section{Sequential Procedures}\label{1.S3}
In this section, we propose both two-stage and purely sequential sampling procedures to estimate the optimal sample size $C$, ensuring that the resulting bounded-width confidence interval achieving the desired confidence level $(1-\alpha)$ asymptotically.
\subsection{Purely Sequential Procedure}
In the purely sequential procedure, we begin by selecting a sample of $m_s$ clusters from each stratum $s$. Thus, the total number of clusters sampled in the first stage, often called the pilot stage, is $m=\sum^S_{s=1}m_s$. Thus the total pilot cluster size is $m$ and $m_s$ is the pilot cluster size from the $s^{\text{th}}$ stratum. Note as per the design, each cluster is subdivided into two sub-strata. From each of the sub-stratum, we randomly draw $k$ households. Using the data obtained is the pilot stage, we compute the estimator $V_m^2$ of the asymptotic variance $\xi^2$, and proceed to evaluate the following stopping rule.
\[N=N(\omega)(\leq H) \text{ is the smallest integer } n(\geq m) \text{ such that for  }\delta>0.\]
\begin{eqnarray}\label{1.5}
n\geq\frac{4z_{{\alpha/2}}^2}{\omega^2}\left(V_n^2+\frac{1}{n^\delta}\right)=\widehat{C} \text{ and } n_s\geq\widehat{C}_s=\widehat{C}a_s, \text{ for all } s,
\end{eqnarray}
where the term  $1/n^\delta$ is included, because the stopping rule \eqref{1.5} will be satisfied even for very small sample sizes in the absence of term $1/n^\delta$ due to the possibility of small value of $V_n^2$ at the early stages, and $V_n^2+1/n^\delta$ is still a consistent estimator of $\xi^2$.

If the above condition is satisfied, then we stop the procedure and our pilot size becomes the final sample size. If not, then we choose $m'(\geq1)$ additional number of clusters from each stratum having $n_s<\widehat{C}_s$ and from each selected additional cluster we randomly select $k$ households. Using this dataset, we estimate $\xi^2$ and check our stopping rule. We repeat this process until the stopping rule is met.
The final cluster size from each stratum $s$ is $N_s=Na_s$. Based on final cluster size $N$, the $100(1-\alpha)\%$ bounded-width confidence interval for Gini index $G_X$ is given by
\begin{eqnarray}\label{1.6}
\Bigg(\widehat{G}_N-z_{\alpha/2}\frac{V_N}{\sqrt{N}},\widehat{G}_N+z_{\alpha/2}\frac{V_N}{\sqrt{N}}\Bigg).
\end{eqnarray}
The pilot cluster size is computed as follows:

Darku et al. (\citeyear{DARKU}) computed the pilot cluster size for $\delta=1$. Through same procedure, for any $\delta>0$, we have
\[n\geq\frac{4z_{{\alpha/2}}^2}{\omega^2}\left(V_n^2+\frac{1}{n^\delta}\right)\geq\frac{4z_{{\alpha/2}}^2}{\omega^2}\left(\frac{1}{n^\delta}\right)\implies n^{\delta+1}\geq\frac{4z_{{\alpha/2}}^2}{\omega^2}\implies n\geq\left(\frac{2z_{{\alpha/2}}}{\omega}\right)^{2/\left(\delta+1\right)}.\]
Therefore, the number of total clusters to be selected during the pilot stage is
\[m=min\left\{H,max\left\{2,\Bigg\lceil\left(\frac{2z_{{\alpha/2}}}{\omega}\right)^{2/\left(\delta+1\right)}\Bigg\rceil\right\}\right\},\]
whereas the number of clusters to be selected from $s^{th}$ stratum during the pilot stage is
\begin{eqnarray}\label{1.Pilot}
m_s=min\left\{H_s,max\left\{2,\Bigg\lceil\left(\left(\frac{2z_{{\alpha/2}}}{\omega}\right)^{2/\left(\delta+1\right)}\right)\times a_s\Bigg\rceil\right\}\right\}.
\end{eqnarray}
\subsection{Two-Stage Procedure}
In two-stage procedure, at the first stage, we have to take a pilot cluster sample of size $t_s$ from each of the $s^{th}$ stratum, where $t_s$ is the same as $m_s$ given by equation \eqref{1.Pilot}. Based on a pilot cluster size of $t=\sum_{s=1}^{S}t_s$ collected at the pilot stage, the total final cluster size to be sampled  from all strata is
\begin{eqnarray}\label{1.7}
Q = \min \left\{ H, \max \left\{ t, \Bigg\lceil \frac{4z^2_{\alpha/2}}{\omega^2 }V_t^2 \Bigg\rceil \right\} \right\} = \min \{ H, Q^* \}.
\end{eqnarray}
Hence, the final cluster size from each stratum $s$ is
\[Q_s=min\{H_s,\lceil Qa_s\rceil\}.\]
Therefore, in the second stage, we have to select $Q_s-t_s$ number of clusters from each stratum $s$. Then, we have to collect data from randomly chosen $k$ households from each sub-stratum of each selected cluster. Based on the collected data, the $100(1-\alpha)\%$ confidence interval for Gini index $G_X$ is given by
\begin{eqnarray}\label{1.8}
\Bigg(\widehat{G}_Q-z_{\alpha/2}\frac{V_Q}{\sqrt{Q}},\widehat{G}_Q+z_{\alpha/2}\frac{V_Q}{\sqrt{Q}}\Bigg).
\end{eqnarray}
\section{Results}\label{1.S4}
In this section, we discuss the theoretical results related to the sequential procedures discussed in Section \ref{1.S3}. Under the assumptions provided in Section \ref{1.S2}, we prove the asymptotic efficiency and asymptotic consistency properties of our procedure. Now, we provide the theorems related to the asymptotic properties of our procedure. Theorem~\ref{1.TH1} is related to the asymptotic efficiency property, whereas Theorem~\ref{1.TH2} is related to the asymptotic consistency property.
\begin{theorem}\label{1.TH1}
If E$[V_n^2]$ exists and for each stratum $s$ of the population, the total number of clusters $H_s$ is large enough, then the following hold. \\ \textit{(i)} Let $N$ be the stopping rule for Purely Sequential Procedure defined in \eqref{1.5}, then as $\omega\rightarrow$ 0,
\[\frac{N}{C}\xrightarrow{a.s.}1,\text{ and } E\left(\frac{N}{C}\right)\rightarrow1.\]
\textit{(ii)} Let $Q$ be the stopping rule for the Two-Stage Procedure defined in \eqref{1.7}, then as $\omega\rightarrow$0,
\[\frac{Q}{C}\xrightarrow{a.s.}1,\text{ and } E\left(\frac{Q}{C}\right)\rightarrow1.\]
\end{theorem}
\textbf{Proof. } The proof of the theorem is given in Appendix.
\begin{theorem}\label{1.TH2} Under the assumptions provided in Section \ref{1.S2}, let $N$ be the stopping rule for Purely Sequential Procedure defined in \eqref{1.5}, then as $\omega\rightarrow$0,
\[P\left\{\widehat{G}_{N}-z_{\frac{\alpha}{2}}\frac{V_N}{\sqrt{N}}<G_X<\widehat{G}_{N}+z_{\frac{\alpha}{2}}\frac{V_N}{\sqrt{N}}\right\}\rightarrow1-\alpha.\]
The same will hold if we replace $N$ with $Q$, where $Q$ is the stopping rule for the Two-Stage Procedure defined in \eqref{1.7}.
\end{theorem}
\textbf{Proof. } The proof of the theorem is given in the Appendix.
\section{Simulation}\label{1.S5}
In this section, we conduct a simulation study under a complex household survey design to illustrate the performance of both the purely sequential and the two-stage procedures in constructing a $100(1-\alpha)\%$ confidence interval for the Gini index, with the constraint that the interval width does not exceed a pre-specified bound $\omega$.

To illustrate the effectiveness of these procedures, we compare the outcomes such as the achieved sample size and coverage probability with their corresponding values under the ideal scenario where population characteristics are assumed to be known.

For simulation study, construction of pseudo population is given in Section 1 of supplementary material.
\subsection{Results for Purely Sequential Procedure}
In this subsection, we carry out a simulation study to explore the properties of our purely sequential procedure. After collecting the pilot sample from the pseudo population, the estimator $V_m^2$ of $\xi^2$ is computed and checked the stopping rule \eqref{1.5}. If satisfied, our pilot cluster size ($m$) is our final cluster size. If not, then we choose $m'\;(=1)$ additional cluster(s) from each stratum $s$ with $n_s<\widehat{C}_s$. We repeat this process until the stopping rule is met. Based on final cluster size $N$, we estimate $\xi^2$ using $V_N^2$, population Gini index $G_X$ using $\widehat{G}_N$, respectively, and construct $100(1-\alpha)\%$ confidence interval for the Gini index $G_X$ as given in \eqref{1.6}.

This procedure is repeated 5000 times. For each replication, we estimate $\xi^2$, the Gini index $G_X$, and construct the corresponding confidence interval. The empirical results for the purely sequential procedure are summarized in Table \ref{1.T5} for the setting $\alpha=0.05, \omega=0.015$.
 

The first column in Table~\ref{1.T5} lists the distribution of the household's monthly income. The second column reports the average final number of clusters ($\bar{N}$) along with the standard deviation ($sd(N)$). The third column presents the average estimated Gini index ($\bar{\widehat{G}}_N$) based on the final sample size. The fourth column shows the ratio of the average final number of clusters to the optimal number of clusters ($\bar{N}/C$). The fifth column displays the average ratio of the estimated variance to the true variance ($\bar{V}^2_N/\xi^2$). The sixth column provides the coverage probability ($\widehat{p}$) of the $5000$ constructed confidence intervals along with its standard error ($se(\widehat{p})$). Column seven gives the average length of these intervals ($\bar{w}_N$) and its standard deviation ($sd({w}_N)$). The final column reports the percentage of intervals whose widths exceed the pre-specified bound $\omega$, denoted by $\widehat{P}(w_N>\omega)$ along with its standard error ($se(\widehat{P}(w_N>\omega))$).

\begin{table}[hbt!]
\caption{Results for $\alpha=0.05$, $\omega=0.015$, and $\delta=2$.}
\centering
\begin{tabular}{ccccccccc}
		\toprule
		Distribution  & $\overline{N}$
		& $\overline{\widehat{G}}_N$  &$\frac{\overline{N}}{C}$&$ \frac{\overline{V_N^2}}{\xi^2}$&$\widehat{p} $&$\overline{w}_N$ & $\widehat{P}(w_N>\omega)$\\
        & $(sd({N}))$&   &&&$(se(\widehat{p}))$&$(sd({w_N}))$& $se(\widehat{P}(w_N>\omega))$\\
		\midrule
		 Gamma &888.176&0.3311&0.9923&0.9909&0.9486&0.01498 &0\\
         (2.649,0.84)&(59.25999)&&&&(0.0031)&(0.000009)&(0)\\
         Pareto&442.15&0.1109&0.9367&0.9338&0.9384&0.0149&0\\
        (20000,5)&(119.2319)&&&&(0.0034)&(0.00002)&(0)\\
		lognormal&1008.81&0.3090&0.9910&0.9897&0.9514&0.0150&0\\
	(2.185,0.562)&(109.1879)&&&&(0.0030)&(0.000009)&(0)\\
	
		\bottomrule
	\end{tabular}
    \label{1.T5}
\end{table}

The results in Table \ref{1.T5} demonstrate that the average widths of the confidence intervals are smaller than the pre-assigned thresholds under the purely sequential procedure. The last column further confirms that, in each replication, the interval width is consistently below the specified bound. The estimated coverage probabilities are close to the nominal level of $100(1-\alpha)\%$. Additionally, the ratio of the average final number of clusters to the optimal number of clusters is approximately 1. These findings collectively support the theoretical properties of the purely sequential procedure discussed in Section \ref{1.S4}.
{Further, simulation results corresponding to another setting ($\alpha=0.05$, $\omega=0.01$) is given in Section 2 of supplementary material.}
\subsection{Results for Two Stage Procedure}
In this subsection, we carry out a simulation study to explore the properties of our two-stage procedure. After collecting the pilot sample from the pseudo population, the estimator $V_t^2$ of $\xi^2$ is computed and the value of the final cluster size is computed using the stopping rule \eqref{1.7}. If $Q=t$, our pilot cluster size $(t)$ is our final cluster size. If $Q>t$, then at the second stage, we have to sample $Q_s-t_s$ number of clusters from each stratum $s$. Based on final cluster size $Q$, we estimate $\xi^2$ using $V_Q^2$, population Gini index $G_X$ using $\widehat{G}_Q$, respectively, and construct $100(1-\alpha)\%$ confidence interval for the Gini index $G_X$ as given in \eqref{1.8}.

This procedure is repeated 5000 times. For each replication, we estimate $\xi^2$, population Gini index $G_X$, and construct the corresponding confidence interval. The empirical results for the two-stage procedure are summarized in Table \ref{1.T7} for the setting $\alpha=0.05, \omega=0.015$.

The first column in Table \ref{1.T7} lists the distribution of the household's monthly income. The second column reports the average final number of clusters($\bar{Q}$) along with the standard deviation $(sd({Q}))$. The third column presents the average estimated Gini index ($\bar{\widehat{G}}_Q$) based on the final sample size. The fourth column shows the ratio of the average final number of clusters to the optimal number of clusters ($\bar{Q}/C$). The fifth column displays the average ratio of the estimated variance to the true variance ($\bar{V}^2_Q/\xi^2$). The sixth column provides the coverage probability ($\widehat{p}$) of the 5000 constructed confidence intervals along with its standard error ($se(\widehat{p}$)). Column seven gives the average length of these intervals ($\bar{w}_Q$) and its standard deviation $(sd({w_Q}))$. The final column reports the percentage of  intervals whose widths exceed the pre-assigned bound $\omega$, denoted by $\widehat{P}(w_Q>\omega)$ along with its standard error ($se(\widehat{P}(w_Q>\omega))$).

\begin{table}[hbt!]
\caption{Results for $\alpha=0.05$, $\omega=0.015$, and $\delta=2$.}
\label{1.T7}
\centering
\resizebox{\textwidth}{!}{%
	\begin{tabular}{ccccccccccc}
   
		\toprule

		Distribution &$\overline{Q}$
		& $\overline{\widehat{G}}_Q$  &$\frac{\overline{Q}}{C}$&$ \frac{\overline{V_Q^2}}{\xi^2}$&$\widehat{p} $&$\overline{w}_Q$ & $\widehat{P}(w_Q>\omega)$\\
        
         & $(sd({Q}))$&   &&&$(se(\widehat{p}))$&$(sd({w_Q}))$& $se(\widehat{P}(w_Q>\omega))$\\
		\midrule
		 Gamma&874.62&0.3311&0.9772&0.9904&0.954&0.0156&0.586\\
          (2.649,0.84)&(270.03)&&&&(0.0029)&(0.0022)&(0.0070)\\
	Pareto&449.305&0.1111&0.9519&0.9514&0.936&0.0169&0.623\\
        (20000,5)&(382.101)&&&&(0.0035)&(0.0053) &(0.0068)\\

		lognormal&983.932&0.3089&0.9665&0.9930&0.946&0.0162 & 0.6394\\
        (2.185,0.562)&(515.263)&&&&(0.0032)&(0.0032)&(0.0068)\\
		\bottomrule
	\end{tabular}%
    }
\end{table}

The results in Table \ref{1.T7} demonstrate that the estimated coverage probabilities are close to the nominal level of $100(1-\alpha)\%$. Additionally, the ratio of the average final number of clusters to the optimal number of clusters is approximately 1. These findings collectively support the theoretical properties of the two-stage procedure discussed in Section \ref{1.S4}. The width of the confidence interval may exceed the pre-assigned bound $\omega$.

{For simulation results corresponding to another setting ($\alpha=0.05$, $\omega=0.01$), see Section 2 of supplementary material.}
\subsection{Discussion}
In this subsection, we discuss some points based on our simulation results for purely sequential as well as for two-stage procedures. We discuss the proportion of confidence intervals whose width exceeds the pre-assigned width $\omega$. Also, we  discuss the empirical distribution of the final sample sizes for both procedures.
\subsubsection{Width of the Confidence Intervals}
For purely sequential procedure, simulation results in Table \ref{1.T5} show that the width of the confidence intervals is less than the pre-assigned width $\omega$. This happens because, by stopping rule \eqref{1.5} of the purely sequential procedure, we have
\[N\geq\frac{4z_{{\alpha/2}}^2}{\omega^2}\left(V_N^2+\frac{1}{N^\delta}\right)\geq\frac{4z_{{\alpha/2}}^2}{\omega^2}V_N^2.\]
This implies
\[\frac{2z_{\alpha/2}V_N}{\sqrt{N}}\leq\omega.\] 
Thus, the final cluster size obtained using a purely sequential procedure guarantees the bounded-width $100(1-\alpha)\%$ confidence interval.

The simulation results for the two-stage procedure, reported in Table \ref{1.T7}, indicate that the width of the confidence interval exceeds the pre-assigned threshold $\omega$. In this procedure, data are collected in two stages. Due to the limited information about $\xi^2$ at the first stage, the estimator $V_t^2$ exhibits high variability, which in turn increases the variability of the final estimator $V_Q^2$. As a result, a considerable proportion of the constructed confidence intervals exceed the desired width. To further investigate this behavior, we examine the proportion of intervals exceeding the specified width under varying pilot cluster sizes. Specifically, we conduct simulations for different values of $\delta$ (i.e., different pilot sample sizes) across the specified income distributions. The corresponding results are summarized in Tables \ref{1.T9}–\ref{1.T11}.
\begin{table}[ht]
\centering
\begin{minipage}{0.45\linewidth}
\centering
\caption{Results for Two Stage procedure for Gamma(2.649, 0.84).}
\label{1.T9}
\resizebox{\textwidth}{!}{%
\begin{tabular}{cccc}
\hline
\toprule
		&$\alpha=0.05$& $\omega=0.015$&\\
         \hline
		$\delta$& Pilot&$\frac{\overline{Q}}{C}$&$\widehat{P}(w_Q>\omega)$ \\
        &Cluster Size&&$se(\widehat{P}(w_Q>\omega))$\\
		\midrule
		2&42&0.9772&0.586 \\
        &&&(0.0070)\\

		1&262&0.9953&0.513 \\
        &&&(0.0071)\\

		0.8&486&0.9970&0.4742 \\
        &&&(0.0071)\\

		0.65&852&1.007&0.255\\
        &&& (0.0062)\\
		\bottomrule
\end{tabular}%
}
\end{minipage}
\hfill
\begin{minipage}{0.45\linewidth}
\centering
\caption{Results for Two Stage procedure for Pareto(20000, 5).}
\label{1.T10}
\resizebox{\textwidth}{!}{%
\begin{tabular}{cccc}
\toprule
		 &$\alpha=0.05$& $\omega=0.015$&\\
         \hline
		$\delta$& Pilot&$\frac{\overline{Q}}{C}$&$\widehat{P}(w_Q>\omega)$ \\
       &Cluster Size&&$se(\widehat{P}(w_Q>\omega))$\\
		\midrule
		2&42&0.9519&0.623\\
		    &&& (0.0068)\\

		1.5&86&0.9869&0.5602\\
		    &&& (0.0070)\\

		1&262&1.0075&0.3632\\
		    &&& (0.0068)\\
	
		0.85&412&1.0541&0.1282 \\
            &&&(0.0047)\\
		\bottomrule
\end{tabular}%
}
\end{minipage}
\end{table}
\begin{table}[hbt!]
 \centering
 \caption{Results for Two Stage procedure for lognormal(2.185, 0.562).}
 \label{1.T11}
        \begin{tabular}{cccc}
		\toprule
         &$\alpha=0.05$& $\omega=0.015$&\\
         \hline
		$\delta$& Pilot&$\frac{\overline{Q}}{C}$&$\widehat{P}(w_Q>\omega)$ \\
       &Cluster Size&&$se(\widehat{P}(w_Q>\omega))$\\
		\midrule
		2&42&0.9665&0.6394\\
		&&& (0.0068)\\

		1&262&0.9952&0.5326\\
		&&& (0.0071)\\
	
		0.8&486&1.0009&0.4754\\
		&&& (0.0071)\\
	
		0.62&966&1.0216&0.1702 \\
        &&&(0.0047)\\
        \bottomrule
	\end{tabular}
  \end{table}
The results in Tables \ref{1.T9}–\ref{1.T11} show that, as the pilot sample size increases, the proportion of confidence intervals exceeding the pre-assigned width $\omega$ decreases. This is because larger pilot cluster sizes lead to reduced variability in the estimator $V_t^2$. Notably, the results reveal a substantial drop in the proportion of wider intervals beyond a certain pilot sample size. 
However, in practical survey applications, selecting a large pilot cluster size may not be desirable due to cost and logistical constraints.
 
From these observations, we conclude that the purely sequential procedure consistently yields confidence intervals whose widths are below the pre-assigned threshold $\omega$. In contrast, under the two-stage procedure, the confidence interval width may exceed $\omega$, particularly when the pilot sample size is small.

{A similar discussion for $\alpha=0.01$ is provided in Section 2 of supplementary material. Further, for empirical distribution, one may look at histogram provided in Section 3 of supplementary material.}
\section{Application results for 64th round NSS Data}\label{1.S6}
In this section, we illustrate the purely sequential and two-stage procedures using survey data from the $64$-th round of the National Sample Survey (NSS) for four Indian states: Maharashtra, Uttar Pradesh, West Bengal, and Tamil Nadu. 

Tables \ref{1.T12} and \ref{1.T13} summarize the results for the purely sequential and two-stage procedures, respectively. Each table lists the state (first column) and the total number of clusters ($H$) in the survey data (second column). The third column shows the final number of clusters - $N$ for the purely sequential procedure in Table \ref{1.T12} and $Q$ for the two-stage procedure  in Table \ref{1.T13} along with the corresponding pilot sample size. The fourth column presents the respective Gini index estimate and its standard error. The fifth and sixth columns report the lower and upper confidence limits, and the last column gives the width of the confidence interval.


\begin{table}[hbt!]
\caption{Results for Purely Sequential Procedure $\alpha=0.1$, $\omega=0.02$ and $\delta=1$.}
\label{1.T12}
\centering

	\begin{tabular}{ccccccc}
		\toprule
		
		State& H & N (m)& $\widehat{G}_N$(se($\widehat{G}_N$)) &Lower CI &Upper CI &$w_N$\\
		\midrule
		Maharashtra& 1008 &1004(166) & 0.2913(0.0060)&0.2815 &0.3012 &0.0197\\
		
		Uttar Pradesh& 1262 &653(166) & 0.2653(0.0058
		)&0.2557 & 0.2748&0.0192\\
		
		West Bengal& 878 & 540(166)&0.2806(0.0058) & 0.2711& 0.2901&0.0190\\
		
		Tamil Nadu& 709 & 512(166)&0.2652(0.0057) & 0.2557
		& 0.2746&0.0189\\
		\bottomrule
	\end{tabular}
\end{table}

  
\begin{table}[hbt!]
\caption{Results for Two-Stage Procedure $\alpha=0.1$, $\omega=0.02$ and $\delta=1$.}
\label{1.T13}
\centering

	\begin{tabular}{ccccccc}
		\toprule
		State& H & Q(t)& $\widehat{G}_Q$(se($\widehat{G}_Q$)) &Lower CI &Upper CI &$w_Q$\\
		\midrule
		Maharashtra& 1008 & 822(166)&0.2911 (0.0055) & 0.2821&0.3002 &0.0181\\
	
		Uttar Pradesh& 1262 & 854(166) &0.2612 (0.0051) & 0.2529& 0.2696& 0.0168\\
		West Bengal& 878 & 586(166)&0.2751 (0.0053) &0.2664 &0.2839 & 0.0175\\
		
		Tamil Nadu& 709 & 566(166)&0.2677(0.0055) & 0.2587&0.2767&0.0180 \\
		\bottomrule
	\end{tabular}
\end{table}
The results in Tables~\ref{1.T12}–\ref{1.T13} show that both procedures achieve the desired precision (i.e., a narrow confidence interval with $w_N \leq \omega$ or $w_Q \leq \omega$) for the Gini index by selecting relatively few clusters while still meeting the required confidence level. This is evident for Maharashtra, Uttar Pradesh, West Bengal, and Tamil Nadu. 

Table~\ref{1.T20} presents the results obtained from real survey data under fixed cluster size. The first column lists the states, while the second and third columns report the fixed cluster sizes $n_1$ and $n_2$. Here, $n_1$ corresponds to under-sampling and $n_2$ to oversampling, relative to the cluster size determined through the sequential procedure (see Tables~\ref{1.T12} and~\ref{1.T13}). The fourth and fifth columns show the widths of the confidence intervals associated with $n_1$ and $n_2$, respectively. As observed, under-sampling ($n_1$) leads to interval widths exceeding $0.02$, whereas oversampling ($n_2$) yields intervals narrower than $0.02$, but at a higher sampling cost. In comparison, the sequential procedure achieves the prescribed width with substantially better precision–cost balance, highlighting its practical advantage over fixed cluster size approaches.
\begin{table}[hbt!]
\caption{Real data results under fixed cluster sizes $n_1$ and $n_2$.}
\centering
\begin{tabular}{ccccc}
		\toprule
		State & $n_1$&$n_2$& $w_{n_1}$&$w_{n_2}$\\
        \midrule
        Maharashtra&800&1008&0.0228&0.0198\\
        Uttar Pradesh&500&1000&0.0232&0.0167\\
		 West Bengal&350&700&0.0261&0.0172\\
          Tamil Nadu&320&650&0.0224&0.0154\\
            \addlinespace
	\bottomrule
	\end{tabular}
    \label{1.T20}
\end{table}

\section{Discussion on Sub-Stratum Effect}\label{1.S7}
In this section, we examine effect of introducing sub-stratification on the final number of clusters, as well as the precision of the Gini index estimation. As discussed earlier, the complex household survey design used by \cite{DARKU} involves stratifying the population and forming clusters of households within each stratum. Their design does not incorporate sub-strata within the selected clusters, an approach similar to that adopted by several survey agencies around the world to account for further heterogeneity within clusters. Also, we note that our proposed design introduces an additional layer of sub-stratification by dividing each selected cluster into two sub-strata: affluent and non-affluent households. To investigate the influence of this additional sub-stratification, we simulate data from three income distributions representative of typical household income patterns: Gamma(2.649, 0.84), Pareto(20000, 5), and Lognormal(2.185, 0.562). For each distribution, we evaluate the final number of clusters determined by the procedures outlined in \cite{DARKU}) and compare them with those obtained using our approach via equations \eqref{1.5} and \eqref{1.7}. To ensure a fair comparison, we use the additional term as $1/n^2$ for the sequential procedure proposed by  \cite{DARKU} instead of $1/n$. For clarity of presentation, we refer to our sampling design as “Proposed” and to the design of \cite{DARKU} as “Darku et al.” in Tables \ref{1.T16}-\ref{1.T17}.

Table \ref{1.T16} summarizes the final number of clusters by the purely sequential and the two-stage procedures based on $5000$ simulations. The first column lists the underlying income distributions. The second column presents the average number of clusters under the purely sequential ($\bar{N}$) and two-stage ($\bar{Q}$) procedures using our proposed sub-stratified approach. These are computed via equations \eqref{1.5} and \eqref{1.7}, respectively. The third column report the corresponding averages obtained using the methods of \cite{DARKU}.
\begin{table}[ht]
\centering
\begin{minipage}{0.45\linewidth}
\centering
\caption{Final number of cluster for $\alpha=0.05$, $\omega=0.015$.}
\label{1.T16}
\begin{tabular}{ccc}
		\toprule
		Distribution & $\bar{N}\left(\bar{Q}\right)$ &$\bar{N}\left(\bar{Q}\right)$\\
        &Proposed&Darku et al.\\
        \midrule
		 Gamma&888.176&948.344\\
 
         (2.649,0.84)&(874.62)&(936.8008)\\

        \midrule
		Pareto&442.15&468.7324\\

        (20000,5)&(449.305)&(489.8396)\\

        \midrule
		lognormal&1008.81&1077.854\\

	(2.185,0.562)&(983.932)&(1055.457)\\
	\bottomrule
	\end{tabular}
\end{minipage}
\hfill
\begin{minipage}{0.45\linewidth}
\centering
\caption{Estimate of Gini index and asymptotic variance using 1200 clusters.}
\label{1.T17}
 \resizebox{\textwidth}{!}{%
    \begin{tabular}{ccc}
		\toprule
		Distribution & $\overline{\widehat{G}}_n\left(\overline{V_n^2}\right)$ &$\overline{\widehat{G}}_n\left(\overline{V_n^2}\right)$\\
        &Proposed&Darku et al.\\
        \midrule
		 Gamma&0.331149&0.3311575\\
       
         (2.649,0.84)&(0.01307985)&(0.01393041)\\

        \midrule
		Pareto&0.1112591&0.1112986\\
     
        (20000,5)&(0.00692779)&(0.007426569)\\
	
        \midrule
		lognormal&0.3090566&0.30910412\\
      
	(2.185,0.562)&(0.01489573)&(0.01600393)\\
	\bottomrule
	\end{tabular}%
    }
\end{minipage}
\end{table}
In Table \ref{1.T17}, we compare the performance of both procedures under a fixed number of clusters, set at 1200. For each income distribution, we conducted 5000 simulation runs under this setting. The results are summarized. The first column lists the household monthly income distributions used in the simulations. The second column present the average estimated Gini index and its estimated variance using equations~(\ref{1.3}) and (\ref{V.E.})–(\ref{std2.3}).
The results presented in Table \ref{1.T17} demonstrate that, on average, the Gini index estimates from the proposed design using sub-stratification are closer to the true values than those obtained using the design used in \cite{DARKU}. Additionally, the proposed design yields consistently lower estimates of the asymptotic variance, indicating that accounting for sub-stratification improves the precision of the Gini estimator. Together, Tables \ref{1.T16}–\ref{1.T17} underscore the benefits of incorporating sub-stratification—yielding more accurate estimates with reduced variability and lower sampling costs compared to designs that do not account for this structure.
\section{Concluding Remarks}\label{1.S8}
The Gini index serves as an important indicator of economic inequality and plays a vital role in informing policy decisions. This article proposes constructing a bounded-width confidence interval for the Gini index using a complex household survey design that stratifies the population, divides each stratum into clusters, and further classifies households within clusters into affluent and non-affluent groups.

In this article, we develop two procedures using this design - a purely sequential procedure and a two-stage procedure, to determine the optimal cluster size needed to construct a bounded-width confidence interval for the Gini index without assuming any specific income distribution. Unlike existing approaches, our design accounts for sub-stratification within clusters, aligning with real survey practices. We show that, under mild regularity conditions, both methods ensure that when the desired width is sufficiently narrow, the final sample cluster size is close to optimal and the confidence interval achieves the required coverage probability. Simulation results confirm these asymptotic properties. The purely sequential procedure yields smaller standard errors in sample sizes and is therefore preferred, while the two-stage procedure offers a simpler implementation in settings where only limited sampling stages are feasible due to logistical or cost considerations.
\section{Appendix}
\renewcommand{\theequation}{A.\arabic{equation}} 
{Below we give a list of lemmas whose proofs are given in Section 4 of the supplementary material.}
\begin{lemma}\label{1.L1} If $\xi^2$ is the asymptotic variance as defined in section \ref{1.S2} and $V_t^2$ is its estimator, then $E(V_t^2)\rightarrow\xi^2$ as $t\rightarrow\infty$.
\end{lemma}
\begin{lemma}\label{1.L2} If $\xi^2$ is the asymptotic variance as defined in section \ref{1.S2} and $V_n^2$ is its estimator, then $V_n^2\xrightarrow{a.s.}\xi^2$ as $n\rightarrow\infty$.
\end{lemma}
\begin{lemma}\label{1.L3}$\widehat{\mu}_n=\left(\sum_{s=1}^{S}\sum_{c_s=1}^{n_s}\sum_{b_{c_s}=1}^{2}\sum_{h=1}^{k}w_{sc_sb_{c_s}h}x_{sc_sb_{c_s}h}\right)$ satisfies uniform continuity in probability condition.
\end{lemma}
\begin{lemma}\label{1.L5}For each $p\in[0,1]$, $\sqrt{n}(F(z(p))-\widehat{F}_n(\widehat{z}_n(p)))\rightarrow0 \text{ as } n\rightarrow\infty.$
\end{lemma}
\begin{lemma}\label{1.L6}For any $\epsilon>0, \text{ there exists }M>0$ such that \[P\{|\sqrt{n}(F(z(p))-\widehat{F}_n(z(p)))|\geq M\}<\epsilon \text{ for all }n\in\mathbb{N}.\]
\end{lemma}
\begin{lemma}\label{1.L7}For each $p\in[0,1]$, $\sqrt{n}\left\{\widehat{z}_n(p)-z(p)+\frac{\widehat{F}_n(z(p))-F(z(p))}{f(z(p))}\right\}\xrightarrow{P}0 \text{ as } n\rightarrow\infty.$
\end{lemma}
\begin{lemma}\label{1.L8}Let $\{R_n\}$ be a sequence of random variables such that $\sqrt{n}R_n\xrightarrow{P}0$, then $\{R_n\}$ will satisfy uniform continuity in probability condition.
\end{lemma}
\begin{lemma}\label{1.L9}For each $p\in[0,1]$, $\widehat{z}_n(p)$ satisfies uniform continuity in probability condition.
\end{lemma}
\begin{lemma}\label{1.L10}$\int_{0}^{1}\widehat{\alpha}_n(p)dp$ satisfies uniform continuity in probability condition.
\end{lemma}
\begin{lemma}\label{1.L11}Under assumptions \ref{1.AN1}-\ref{1.AN10}. For each $p\in[0,1], \text{ as }n\rightarrow\infty$  \[\sqrt{n}(\widehat{\mu}_{n}-\mu)\xrightarrow{d}N(0,V_1),\text{ and }\sqrt{n}(\widehat{\alpha}_{n}(p)-\alpha(p))\xrightarrow{d}N(0,V_{2,p}),\text{for some } V_1>0\text{ and }V_{2,p}>0.\]
\end{lemma}
\begin{lemma}\label{1.L12}If $N$ is the non-negative integer variable and $C$ is the optimal number of clusters such that $\frac{N}{C}\xrightarrow{P}1  $, then under assumptions \ref{1.AN1}-\ref{1.AN10}, for each $p\in[0,1]$, as $\omega\rightarrow0$ \\
$\textit{(i) } \sqrt{N}(\widehat{\mu}_{N}-\mu)\xrightarrow{d}N(0,V_1).$\hspace{4cm}
$\textit{(ii) } \sqrt{N}(\widehat{\alpha}_{N}(p)-\alpha(p))\xrightarrow{d}N(0,V_{2,p}).$\\
$\textit{(iii) } (\widehat{\mu}_{N}-\mu)\xrightarrow{a.s.}0\text{ and }(\widehat{\alpha}_{N}(p)-\alpha(p))\xrightarrow{a.s.}0.$\hspace{0.8cm}
$\textit{(iv) } \int_{0}^{1}(\widehat{\alpha}_{N}(p)-\alpha(p))dp\xrightarrow{P}0.$
\end{lemma}
\noindent{\bf Proof of Theorem~\ref{1.TH1} }\textit{(i):} 
From stopping rule \eqref{1.5}, we get,
\[ \left(\frac{2z_\frac{\alpha}{2}V_N}{\omega}\right)^2\leq N\leq m+\left(\frac{2z_\frac{\alpha}{2}}{\omega}\right)^2\left(V_{N-1}^2+\frac{1}{(N-1)^{\delta}}\right).\]
By dividing the above inequality by $C$, we get,
\begin{eqnarray}\label{1.N/C}
\left(\frac{V_N}{\xi}\right)^2\leq \frac{N}{C}\leq \frac{m}{C}+\left(\frac{1}{\xi}\right)^2\left(V_{N-1}^2+\frac{1}{(N-1)^{\delta}}\right).
\end{eqnarray}
Using Lemma \ref{1.L2} and \cite{GUT}, we have 
\begin{eqnarray}\label{arna1}
V_N^2\xrightarrow{a.s.}\xi^2.
\end{eqnarray}
Since ${m}/{C}\rightarrow0$ as $\omega\rightarrow0$, by equations \eqref{1.N/C} and \eqref{arna1}, we conclude that ${N}/{C}\xrightarrow{a.s.}1$ as $\omega\rightarrow0$.
Now, we will prove that $E({N})/{C}\rightarrow1$ as $\omega\rightarrow0$. From equation \eqref{1.N/C}, we have,
\[N/C-m/C\leq\frac{1}{\xi^2}\left(\sup_{n\geq m}V_{n}^2+\frac{1}{(m-1)^{\delta}}\right).\]
Since
\[\sup_{n\geq m}V_{n}^2=\sup_{n\geq m}\left(\sum_{s=1}^{S}\frac{n\times n_s}{n_s-1}\sum_{c_s=1}^{n_s}\left(u_{sc_s}-\bar{u}_s\right)^2\right)\leq \sum_{s=1}^{S}H\times H_s\left(\sup_{n\geq m}\left(\frac{1}{n_s-1}\sum_{c_s=1}^{n_s}\left(u_{sc_s}-\bar{u}_s\right)^2\right)\right),\]
by taking Expectations on both sides, we get  
\[E\left(\sup_{n\geq m}V_{n}^2\right)\leq\sum_{s=1}^{S}H\times H_s\times E\left(\sup_{n\geq m}\left(\frac{1}{n_s-1}\sum_{c_s=1}^{n_s}\left(u_{sc_s}-\bar{u}_s\right)^2\right)\right).\]
Using Cauchy-Schwarz inequality and Lemma 9.2.4 of \cite{GHOSH}, we get
\[E\left(\sup_{n\geq m}V_{n}^2\right)\leq\sum_{s=1}^{S}H\times H_s\times \left[E\left(\sup_{n\geq m}\left(\frac{1}{n_s-1}\sum_{c_s=1}^{n_s}\left(u_{sc_s}-\bar{u}_s\right)^2\right)^2\right)\right]^{1/2}\]\[\leq\sum_{s=1}^{S}H\times H_s\times \left[4\times E\left(\left(\frac{1}{m_s-1}\sum_{c_s=1}^{m_s}\left(u_{sc_s}-\bar{u}_s\right)^2\right)^2\right)\right]^{1/2}<\infty.\] 
Since $N/C\xrightarrow{P}1$ as $\omega\rightarrow0$, by dominated convergence theorem, we conclude that
\[E\left(\frac{N}{C}\right)\rightarrow1 \text{ as } \omega\rightarrow0.\]
\textit{(ii)} From the stopping rule in equation \eqref{1.7}, we have
\[ \left(\frac{2z_\frac{\alpha}{2}V_t}{\omega}\right)^2\leq Q\leq t+\left(\frac{2z_\frac{\alpha}{2}}{\omega}\right)^2\left(V_{t}^2\right).\]
By dividing the above inequality by $C$, we get
\begin{eqnarray}\label{1.Q/C}
\frac{V_{t}^2}{\xi^2}\leq \frac{Q}{C}\leq \frac{t}{C}+\frac{V_{t}^2}{\xi^2}.
\end{eqnarray}
Since $t\rightarrow\infty$ as $\omega\rightarrow0$, $V_t^2\xrightarrow{a.s.}\xi^2$. Now, $\frac{t}{C}\rightarrow0$ as $\omega\rightarrow0$. Hence, using equation \eqref{1.Q/C}, we conclude that $\frac{Q}{C}\xrightarrow{a.s.}1$ as $\omega\rightarrow0$. 
Taking expectation both sides, we get
\[ \frac{E\left(V_{t}^2\right)}{\xi^2}\leq \frac{E(Q)}{C}\leq \frac{t}{C}+\frac{E\left(V_{t}^2\right)}{\xi^2}.\]
From Lemma \ref{1.L1}, we have $E(V_t^2)\rightarrow\xi^2$ as $t\rightarrow\infty$. Since $t\rightarrow\infty$ as $\omega\rightarrow0$ and ${t}/{C}\rightarrow0$ as $\omega\rightarrow0$, by using the above inequality, we conclude that
\[\frac{E(Q)}{C}\rightarrow1 \text{ as } \omega\rightarrow0.\]
\noindent\textbf{Proof of Theorem \ref{1.TH2}. }Define $T(u,v)=\frac{u}{v},$ for $v\neq0$. Let $u_N=\int_{0}^{1}\widehat{\alpha}_N(p)dp, v_N=\widehat{\mu}_N, u_0=\int_{0}^{1}{\alpha}(p)dp\text{ and } v_0={\mu}.$ By Taylor series expansion, we get\\
\resizebox{\textwidth}{!}{%
$T(u_N,v_N)=T(u_0,v_0)+\frac{(u_N-u_0)}{v_0}-\frac{u_0(v_N-v_0)}{(v_0)^2}-\frac{(u_N-u_0)(v_N-v_0)}{(v_0+p(v_N-v_0))^2}+\frac{2(u_0+p(u_N-u_0))(v_N-v_0)^2}{(v_0+p(v_N-v_0))^3},$
}
for $p\in(0,1).$ Multiplying the above equation by $\sqrt{N}$, we get
\begin{eqnarray}\label{1.9}
\sqrt{N}(T(u_N,v_N)-T(u_0,v_0))=\sqrt{N}\left(\frac{(u_N-u_0)}{v_0}-u_0\frac{(v_N-v_0)}{(v_0)^2}+R_N\right),
\end{eqnarray}
where
$$R_N=-\frac{(u_N-u_0)(v_N-v_0)}{(v_0+p(v_N-v_0))^2}+\frac{2(u_0+p(u_N-u_0))(v_N-v_0)^2}{(v_0+p(v_N-v_0))^3}.$$ Now, consider
\[\frac{\sqrt{N}(u_N-u_0)}{v_0}-\frac{u_0\sqrt{N}(v_N-v_0)}{(v_0)^2}=\frac{\sqrt{N}(u_N-u_C+u_C-u_0)}{v_0}-\frac{u_0\sqrt{N}(v_N-v_C+v_C-v_0)}{(v_0)^2}\]
\[=\frac{\sqrt{N}(u_N-u_C)}{v_0}-\frac{u_0\sqrt{N}(v_N-v_C)}{(v_0)^2}+\frac{\sqrt{N}(u_C-u_0)}{v_0}-\frac{u_0\sqrt{N}(v_C-v_0)}{(v_0)^2}\]
\[=A+B\times D,\]
where
\[A=\frac{\sqrt{N}(u_N-u_C)}{v_0}-\frac{u_0\sqrt{N}(v_N-v_C)}{(v_0)^2}, B=\sqrt{\frac{N}{C}},\]
\[D=\left(\frac{\sqrt{C}(u_C-u_0)}{v_0}-\frac{u_0\sqrt{C}(v_C-v_0)}{(v_0)^2}\right).\]
Since $\frac{N}{C}\xrightarrow{P}1$, we get from  Lemmas \ref{1.L3} and \ref{1.L10} that $A\xrightarrow{P}0$. By using \cite{BD2007}, we get $D\xrightarrow{d}N(0,\xi^2/4)$. Consequently, by using Slutsky's Theorem, we get
\begin{eqnarray}\label{1.10}
	\frac{\sqrt{N}(u_N-u_0)}{v_0}-\frac{u_0\sqrt{N}(v_N-v_0)}{(v_0)^2}\xrightarrow{d}N(0,\xi^2/4).
\end{eqnarray}
Again,
\[\sqrt{N}R_N=-\frac{(\sqrt{N}(u_N-u_0))(v_N-v_0)}{(v_0+p(v_N-v_0))^2}+\frac{2(u_0+p(u_N-u_0))(\sqrt{N}(v_N-v_0))(v_N-v_0)}{(v_0+p(v_N-v_0))^3},\]which, in view of Lemma \ref{1.L12} and Slutsky's Theorem, gives
\begin{eqnarray}\label{1.11}
	\sqrt{N}R_N\xrightarrow{P}0.
\end{eqnarray}
By using equations \eqref{1.9}, \eqref{1.10}, \eqref{1.11}, and Slutsky's Theorem, we get
\begin{eqnarray}\label{arna}
    \sqrt{N}(T(u_N,v_N)-T(u_0,v_0))\xrightarrow{d}N(0,\xi^2/4).
\end{eqnarray}
Now, we have $T(u_N,v_N)=\frac{u_N}{v_N}=\frac{\int_{0}^{1}\widehat{\alpha}_N(p)dp}{\widehat{\mu}_N}$ and $ T(u_0,v_0)=\frac{u_0}{v_0}=\frac{\int_{0}^{1}{\alpha}(p)dp}{\mu}$.
Further, we have $\widehat{G}_N=1-2(\frac{u_N}{v_N})\text{ and }G_X=1-2(\frac{u_0}{v_0})$. Consequently, we get $\sqrt{N}(\widehat{G}_N-G_X)=2\sqrt{N}(T(u_N,v_N)-T(u_0,v_0))$.
Then, by equation \eqref{arna}, we get
\begin{eqnarray}
	\sqrt{N}(\widehat{G}_N-G_X)\xrightarrow{d}N(0,\xi^2).
\end{eqnarray}
Further, by equation \eqref{arna1}, we have $V_N^2\xrightarrow{P}\xi^2$. Consequently, by equation \eqref{arna}, we get that
\[P\left\{\widehat{G}_{N}-z_{\frac{\alpha}{2}}\frac{V_N}{\sqrt{N}}<G_X<\widehat{G}_{N}+z_{\frac{\alpha}{2}}\frac{V_N}{\sqrt{N}}\right\}=P\left\{\Bigg|\frac{\sqrt{N}(\widehat{G}_N-G_X)}{V_N}\Bigg|<z_\frac{\alpha}{2}\right\}\rightarrow1-\alpha,\]
as $\omega\rightarrow$0. Hence the result.
\bibliographystyle{elsarticle-harv}
\bibliography{name}
\end{document}